\documentclass[a4paper,twocolumn,11pt,accepted=2021-12-31]{quantumarticle}
\pdfoutput=1
\usepackage[utf8]{inputenc}
\usepackage[english]{babel}
\usepackage[T1]{fontenc}
\usepackage{amsmath}
\usepackage{hyperref}

\usepackage{tikz}
\usepackage{lipsum}

\usepackage{amsfonts}
\usepackage{amsthm}
\usepackage{amsmath}
\usepackage{amscd}

\usepackage{t1enc}
\usepackage[mathscr]{eucal}
\usepackage{indentfirst}
\usepackage{graphicx}
\usepackage{graphics}
\usepackage{xcolor}

\usepackage{ulem}

\numberwithin{equation}{section}

\theoremstyle{plain}
\newtheorem{Th}{Theorem}
\newtheorem{Lemma}[Th]{Lemma}

 \theoremstyle{definition}

\newtheorem{Rem}[Th]{Remark}
\newtheorem{?}[Th]{Problem}

\newcommand{\tr}{\operatorname{tr}}

\usepackage{soul}

\begin{document}

\title[On entanglement assistance to a classical channel]{On entanglement assistance to a noiseless classical channel}

\author{P\'eter E.\ Frenkel}
\affiliation{E\"{o}tv\"{o}s Lor\'{a}nd University,
  P\'{a}zm\'{a}ny P\'{e}ter s\'{e}t\'{a}ny 1/C, Budapest, 1117 Hungary \\ and R\'enyi Institute,  Budapest, Re\'altanoda u.\ 13-15, 1053 Hungary}
  \email{frenkelp265@gmail.com}
\orcid{0000-0003-2672-8772}
\thanks{Partially supported  by  ERC Consolidator Grant 648017 and by NKFIH  grants KKP 139502 and  K 124152.}
\author{Mih\'aly Weiner}
\affiliation{Budapest University of Technology and Economics (BME), Department of Analysis,
  H-1111 Budapest
  M\H{u}egyetem rkp. 3--9 Hungary, and
  MTA-BME Lend\"ulet Quantum Information Theory Research Group}
\email{mweiner@math.bme.hu}
\orcid{}
\thanks{Supported by the Bolyai J\'anos Fellowship of the Hungarian Academy of Sciences, the \'UNKP-21-5 New National Excellence Program of the Ministry for Innovation and Technology, by the NRDI grants  K 124152, KH129601 and K132097 and by the NRDI Office within the Quantum Information National Laboratory of Hungary.}
\maketitle

\begin{abstract}
For  a classical channel, neither the Shannon ca\-pa\-city, nor the sum of conditional probabilities corresponding to
the cases of successful transmission   can be increased by the use of shared entanglement, or, more generally, a non-signaling resource.
Yet, perhaps somewhat counterintuitively,
entanglement assistance can help and actually elevate the chances of success even in
a one-way communicational task that is to be completed by a single-shot use of a
noiseless classical channel.

To quantify the help that  a non-signaling resource provides to a noiseless classical channel, one might ask how many extra letters should be added to the alphabet of the  channel  in order to  perform equally well {\it without} the specified non-signaling  resource. As was observed  by Cubitt, Leung, Matthews, and Winter, there is no upper bound on the number of extra letters required for substituting the assistance of a general non-signaling resource to a noiseless one-bit classical channel. In contrast, here we prove that if this resource is a bipartite quantum system in a maximally entangled state, then an extra classical bit always suffices as a replacement.
\end{abstract}

\maketitle

\section{Introduction}

If a certain two-part resource is {\it non-signaling}, then --- essentially by definition --- it cannot be used to exchange messages between its two users.
This applies to  hypothetical resources such as  Popescu--Rohrlich boxes \cite{pr_box} as well as to resources that can be realized within the framework of quantum physics by sharing a two-part quantum system in a prearranged entangled state.
However, as an aid, such a resource might boost the capabilities of an already existing communicational channel between its users.
For example, in  the famous {\it dense coding protocol} \cite{densecoding}, entanglement is used to double the classical capacity of a noiseless quantum channel. When the quantum channel is noisy, entanglement can give an arbitrarily large boost (depending on the channel) \cite{noisyquantum}. This noise-reducing capability of entanglement
has even been  experimentally demonstrated; see e.g.\! \cite{experimental} for recent developments.

The situation  changes somewhat when the channel to be improved is a classical one. This is because it turns out that important quantities such as ``information storability'' --- that is, the sum of conditional probabilities corresponding to the ``output = input'' cases; see  e.g.\!  \cite{geoinfcap} --- or the Shannon capacity of a classical channel cannot be increased\footnote{At least, not in the one-sender-one-receiver case  considered in this paper; see however \cite{multiclassical} for an interesting application of entanglement regarding ``multi-channels''.}
 by the additional use of a non-signaling resource
\cite{zeroerror2}.
However,   entanglement (which  is a particular  type of non-signaling resource) can be used, for example, to increase the zero-error capacity of a noisy classical channel \cite{zeroerror1}. To put it in another way, it can improve the capability of a noisy classical channel to simulate a noiseless one.

One might also consider the inverse problem of using a noiseless classical channel (together with the possible help of a non-signaling resource) to simulate a noisy one. In \cite{zeroerror2}, an example was given for a classical channel which cannot be simulated by (a single use of) a noiseless one-bit classical channel aided only by shared randomness, but which can be simulated by the same channel if assistance, in the form of using a bipartite quantum system prepared in an entangled state, is allowed. Using the concepts introduced in \cite{sigdim}, we may say that the assistance increases the ``signaling dimension'' of our  noiseless one-bit classical channel.

\subsection{Game interpretation}

We might view this simulability question from the point of view of one-way communicational tasks. For example, let us consider the following simple game. We have four boxes, two of them empty, two containing (equal) treasures; thus, there are ${4\choose 2}=6$ possible configurations regarding the positions of the treasures. Each configuration is equally likely, with the actual (secret) configuration revealed
only to Alice, who is allowed to send one classical bit to Bob. After receiving the bit sent by Alice, Bob chooses a box. If it contains a treasure, Alice and Bob win (as a team).

Without a non-signaling resource, relying only on arrangements before the game and a possible use of shared randomness, it is easy to see that the maximum chance of winning can be achieved by a {\it pure} (deterministic) strategy
---  in terms of expected reward, shared randomness is of no use.

   We may assume that upon receiving the bit sent by Alice, depending on its value, Bob points to either box  1 or box  2. With this agreed, Alice can always send a bit value to Bob that makes him point to a treasure box unless the treasures are  in boxes   3 \& 4; that is, Alice and Bob will win
    with a chance of $1-{1}/{6}={5}/{6}$. On the other hand, if during the game Alice and Bob can also make some measurements on a pair of quantum bits (prepared in a maximally entangled state and distributed between them before the start of the game), then there exists a strategy allowing them to  win this game with a chance of $\left({4+\sqrt{2}}\right)/{6}>{5}/{6}$ ---  see the Appendix.

In this protocol involving the use of entanglement, Alice and Bob realize a certain classical channel $\mathcal N\in
C(X\to Y)$ with $|X|=6$ possible inputs and $|Y|=4$ outputs (with the input being the whereabouts of the treasures revealed to Alice and the output being the box chosen by Bob). The use of this channel allows Alice and Bob to win the game with a chance of $\left({4+\sqrt{2}}\right)/{6}$. The fact that $\left({4+\sqrt{2}}\right)/{6} > {5}/{6}$ shows that channel $\mathcal N$ cannot be simulated by a single use of a noiseless classical one-bit channel aided only by shared randomness.

Let us consider, in general, for any pair of (finite) sets $X,Y$, natural number $n$ and non-signaling resource $\omega$ the following:
\begin{itemize}
\item $C_n(X\to Y)$, the set of $X\to Y$ classical channels that can be simulated by a single use of a noiseless classical channel with $n$  different letters (without any other resources),

\item $C_n^{SR}(X\to Y)$, the set of $X\to Y$ classical channels that can be simulated by a single use of a noiseless classical channel with $n$  different letters together with an unlimited source of shared randomness between the sender and receiver,

\item $C_n^\omega(X\to Y)$, the set of $X\to Y$ classical channels that can be simulated by a single use of a noiseless classical channel with $n$  different letters and assistance coming from $\omega$,

\item $C_n^{BQ}(X\to Y)$, the set of $X\to Y$ classical channels that can be simulated by a single use of a noiseless classical channel with $n$  different letters and assistance from any bipartite quantum system (prepared in any state),

\item $C_n^{NS}(X\to Y)$, the set $X\to Y$ classical of channels that can be simulated by a single use of a noiseless classical channel with $n$  different letters and assistance from any non-signaling resource.
\end{itemize}
We postpone  to Section~\ref{prelim}  the precise definition and detailed description of these sets, and -- omitting the $X\to Y$ indication -- note here only that  $C_n^{SR}$ is  the convex hull of $C_n$; 
$C_n^{BQ}$ and $C_n^{NS}$ are convex sets; and
$C_n^{SR}\subseteq C_n^{BQ} \subseteq C_n^{NS}$.

We now generalize the game above and connect the question of (im)possibility of simulations
to advantages in one-way communicational games. We begin  by  describing what we mean by a general one-way communicational game.

Suppose we have a team of players consisting of Alice and Bob. An element $x$ of a (finite) set $X$ is chosen according to a given probability distribution $q$, and revealed to Alice (but not to Bob). At the end of the game, Bob will need to pick an element $y$ of another (finite) set $Y$ and the team receives a reward, but the actual sum of this reward depends {\it both} on his choice $y$ and on the input $x$; it is given by some ``reward-function'' $R:X\times Y\to \mathbb R$. With both the probability distribution $q$ and reward-function $R$ publicly given, we shall consider how the maximum expected reward (achieved by the best team strategy) depends on 
 the allowed forms of communication and non-signaling resources the team can use.

Once a team strategy (using the given channels and resources) is chosen, the conditional probability of Bob choosing $y$, given that Alice receives input $x$, is fixed. Thus, an actual strategy realizes a classical $\mathcal N\in C(X\to Y)$ channel. The expected reward is a linear functional of the realized channel:
$$
\mathbb E({\text {reward}}) = \sum_{x\in X,y\in Y} 
 R(x,y)N(y|x)q(x),
$$
with the functional depending on $R$ and $q$. Since we want to consider all such games, we do not have a restriction on possible reward functions and input probability distributions and thus, for us, the expected reward is just an arbitrary linear functional of the realized channel. It follows that there exists  a one-way communicational game in which the single use of a classical noiseless channel with $n$ different letters together with assistance coming from a non-signaling resource $\omega$ is more advantageous (in terms of maximal expected rewards) than the single use of a classical noiseless channel with $n'$ different letters together with assistance from a non-signaling resource $\omega'$ if and only if
$C_n^\omega(X\to Y)$ is not contained in the convex hull of $C_{n'}^{\omega'}(X\to Y)$ for some $X$ and $Y$. In particular, with assistance coming from a non-signaling resource $\omega$, the use of a classical noiseless channel of $n$ letters is never more advantageous than a single use of an unaided
classical noiseless channel of $m$ different letters if and only if
$$
C_n^\omega(X\to Y) \subseteq C_m^{SR}(X\to Y)
$$
for all possible sets $X$ and $Y$ of input and output symbols.

\subsection{Entanglement vs.\! generic non-signaling resources}

By \cite[Proposition 19]{zeroerror2}, for every $n$ there exists a non-signaling resource $\omega_n$ such that $C_2^{\omega_n}$ is {\it not} contained in $C_n^{SR}$ (for some sets of input and output symbols which from now on we shall omit). Thus, in the sense explained, there is no bound on the advantage that a non-signaling resource can give
to a one-bit classical noiseless channel.  However, the non-signaling resources used in these examples correspond to hypothetical devices that may have no realizations in nature. In  fact, in contrast to the general non-signaling case,  we shall prove in Theorem~\ref{main} that
$
C_2^{\omega}\subseteq C_4^{SR}
$
whenever $\omega$ is realizable by a quantum bipartite system prepared in a maximally entangled state.
The result is general in the sense that it holds without any limit on the size of the quantum system (i.e.\!
 the dimension of the Hilbert space) used. However, we do exploit that the state is a maximally entangled one. Since there are Bell inequalities whose maximal violation occurs in
states which are {\it not} maximally entangled \cite{nonmaximal}, it remains unclear whether
$C_2^{BQ}\subseteq C_4^{SR}$. Nevertheless, we conjecture that this is indeed so; maybe even $C_2^{BQ}\subseteq C_3^{SR}$. Also, it is a  natural guess that $C_n^{BQ}$ should   always  be contained in $C_{n^2}^{SR}$.

If this  is true, it may  turn out to be a particular manifestation of some generic behavior of nature, limiting the ``extra help'' provided by entanglement to be comparable to ``our efforts'':
while generic non-signaling resources can yield  ``unlimited help''  in  one-way communicational tasks, a bipartite  quantum system cannot  give more advantage  than that offered by  allowing a second shot of the employed classical noiseless channel. 
 Similarly to the so-called ``Information Causality'' \cite{infcaus} or the idea of ``No-Hypersignaling'' \cite{sigdim}, this could be viewed as a fundamental principle limiting the non-signaling resources that can appear in nature.

 An interesting parallel can be drawn between the above  and  the {\it graph homomorphism game} \cite{quantumhomo}. With only classical pre-arrangement, Alice and Bob can win that game with certainty only if there exists a homomorphism from graph 1 to graph 2. Allowing them to use any non-signaling resource lets them win with certainty whenever 
 graph 2  has at least one edge. 
   Instead, the use of entanglement
provides a limited help, thus giving rise to the
interesting question: for which pair of graphs is there a ``quantum homomorphism'', i.e.\! a winning strategy using entanglement assistance?

\section{Preliminaries}\label{prelim}

Given two finite sets  $X,Y$ (the ``alphabets''), a classical channel from $X$ to $Y$ is a function
$$\mathcal N: Y\times X \to [0,1]$$ satisfying
$\sum_{y\in Y} \mathcal N(y|x)=1$ for all $x\in X$. We interpret the value $\mathcal N(y|x)$ as the probability of the channel producing the output $y$ given that the input is $x$, and we denote by $C(X\to Y)$ the set of all classical channels from $X$ to $Y$.

For a natural number $n$, set $[n]\equiv  \{1,\ldots, n\}$.
We shall say that $\mathcal N\in C(X\to Y)$ can be realized
by a single use of a noiseless classical channel with $n$
different letters if there exists a pair of {\it encoding}
and {\it decoding}, i.e., channels
$\mathcal N_{enc}\in C(X\to [n])$ and
$\mathcal N_{dec}\in C([n]\to Y)$
such that
$$
\mathcal N(y|x) = \sum_{r=1}^n
\mathcal N_{dec}(y|r) \, \mathcal N_{enc}(r|x)
$$
for all $x\in X$ and $y\in Y$. We denote the set of all such channels $\mathcal N$ by $C_n(X\to Y)$.

Using a source of randomness shared between the sender and receiver, it is possible to mix different encoding--decoding strategies. Thus, the set of
classical channels $C_n^{SR}(X\to Y)$ realizable by a single use of a noiseless classical channel with $n$ different letters aided by an unlimited source of shared randomness is simply the convex hull of $C_n(X\to Y)$. We note that $\mathcal N\in C_n^{SR}(X\to Y)$ if and only if the stochastic matrix
$A$ with entries $a_{i,j}:=\mathcal N(y(i)|x(j))$, where
$x:[l]\to X$ and $y:[k]\to Y$ are some bijections enumerating the $l=|X|$ and $k=|Y|$ elements of $X$ and $Y$,
is a convex combination of stochastic matrices with at most $n$
nonzero rows. (Throughout 
 this paper, by a  {\it stochastic} matrix we mean a matrix with nonnegative entries
whose columns sum to $1$; i.e., $A=(a_{i,j})_{i,j}$ is a stochastic matrix if $a_{i,j}\geq 0$ for all $i$ and $j$, and $\sum_{i} a_{i,j}=1$ for all $j$.)

In our context, a two-part resource $\omega$ is just a classical channel with two inputs and two outputs; i.e.\! an element of $C(X_1\times X_2\to Y_1\times Y_2)$, where $X_1,X_2,Y_1,Y_2$ are some finite sets. We say that $\omega$
is {\it non-signaling} if for any $x_1,x'_1\in X_1$
and $(x_2,y_2)\in X_2\times Y_2$, we have
$$
\sum_{y_1\in Y_1} \omega(y_1,y_2|x_1,x_2) =\sum_{y_1\in Y_1} \omega(y_1,y_2|x'_1,x_2),
$$
i.e., if the choice of the input at access point  1 does not affect the outcome probabilities at access point  2, and, further,  the same holds in the other direction as well:
$$
\sum_{y_2\in Y_2} \omega(y_1,y_2|x_1,x_2) =\sum_{y_2\in Y_2} \omega(y_1,y_2|x_1,x'_2)
$$
for any $x_2,x'_2\in X_2$ and $(x_1,y_1)\in X_1\times Y_1$.

A channel $\mathcal N\in C(X\to Y)$ can be realized by a single use of a noiseless classical channel with $n$ different letters
assisted by a non-signaling resource $\omega$ if there exist
\begin{itemize}
\item a coding  $\mathcal N_{in1}\in C(X\to X_1)$ for the sender  to  select an input for her part of the resource,
\item an encoding $\mathcal N_{enc}\in C(X\times Y_1\to [n])$ for the sender to select (in light of the response of the resource) the message  to be sent,
\item  a coding  $\mathcal N_{in2}\in C([n]\to X_2)$ for the receiver to select an input for his part of the resource,
\item a decoding  $\mathcal N_{dec}\in C([n]\times Y_2\to Y)$ for the receiver to select the output,
\end{itemize}
such that for all $x\in X$ and $y\in Y$,  the transition  probability $\mathcal N(y|x) $  is the sum of all products
\begin{align*}
\mathcal N_{dec}(y|r,y_2)
\mathcal N_{in2}(x_2|r)\cdot \\ \cdot\mathcal N_{enc}(r|x,y_1)
\omega(y_1,y_2|x_1,x_2) \mathcal N_{in1}(x_1|x)
\end{align*}
for 
$r\in[n]$ and $
(x_1,x_2,y_1,y_2)\in X_1\times X_2\times Y_1\times Y_2$. We denote by $C_n^\omega(X\to Y)$ the set of all such channels and by $C_n^{NS}(X\to Y)$ the union of these sets taken over all non-signaling resources  $\omega$. Note that this latter set is automatically convex; this is because shared randomness is also a particular case of a non-signaling resource.

 On a Hilbert space $\mathcal H$, the algebra of bounded linear operators is denoted by $\mathcal B(\mathcal H)$.
Recall that  a  {\it partition of unity} (also known as a {\it positive operator valued measure}; POVM)
on  a  Hilbert space is a collection
of positive semidefinite operators summing to the identity operator $\mathbf 1$.

Let $\mathcal H_A$ and $\mathcal H_B$ be two (complex) Hilbert spaces and
$\rho$ a {\it density operator} on $\mathcal H_A\otimes \mathcal H_B$; i.e.\! a positive semidefinite operator with
$\tr\rho=1$. A two-part resource $\omega\in C(X_A\times X_B\to Y_A\times Y_B)$ is realizable by the use of a bipartite quantum system (with parts corresponding to the spaces $\mathcal H_A$ and $\mathcal H_B$) in state $\rho$ if, for each $x_a\in X_A$ and $x_b\in X_B$, there exist a  partition of unity $\left(F^{(x_a)}_{y_a}\right)_{y_a\in Y_A}$ on $\mathcal H_A$
 and a partition of unity $\left(E^{(x_b)}_{y_b}\right)_{y_b\in Y_B}$ on $\mathcal H_B$ such that
$$
\omega(y_a,y_b|x_a,x_b) = \tr \rho\left(F^{(x_a)}_{y_a}\otimes E^{(x_b)}_{y_b}\right)
$$
for all $(x_a,x_b,y_a,y_b)\in X_A\times X_A\times Y_B\times Y_B$. We note that such a resource is automatically non-signaling, and introduce $C_n^{BQ}(X\to Y)$ as
the union of the sets $C_n^\omega(X\to Y)$ with $\omega$ ranging over all non-signaling resources realizable by the use of some bipartite quantum system prepared in some state.

For a bounded linear operator $Z$ on $\mathcal H_A$, consider the bounded linear operator
\begin{equation}
\label{defofphi}
\Phi_\rho(Z) \equiv \tr_A \rho (Z\otimes \mathbf 1)
\end{equation} on $\mathcal H_B$,
where $\tr_A$ denotes the {\it partial trace} corresponding to $\mathcal H_A$. It is easy to check that
this is well defined and the linear map $\Phi_\rho:\mathcal B(\mathcal H_A)\to\mathcal B(\mathcal H_B)$
 is {\it positive}: if $Z\geq \mathbf 0$, then
$\Phi_\rho(Z)\geq \mathbf 0$.
Let us now introduce, in the previous construction of the non-signaling resource $\omega$, the operator
$$
\beta^{(x_a)}_{y_a}\equiv \Phi_\rho\left(F^{(x_a)}_{y_a}\right),
$$
 i.e.\! the (subnormalized) conditional state of the second subsystem given that the measurement of $x_a$ on part one gave result $y_a$.
Then, for each $x_a\in X_A$, the operators
$\left(\beta^{(x_a)}_{y_a}\right)_{y_a\in Y_a}$ form a {\it positive decomposition} of $\rho_B\equiv \tr_A\rho$; i.e., $\beta^{(x_a)}_{y_a} \geq \mathbf 0$ for all
$x_a$ and $y_a$, and
$$
\sum_{y_a\in Y_a} \beta^{(x_a)}_{y_a} = \rho_B\equiv \tr_A\rho
$$
for all $x_a$. With these newly introduced operators, we can express $\omega$ as follows:
\begin{equation}
\label{fromBobspointofview}
\omega(y_a,y_b|x_a,x_b) = \tr E^{(x_b)}_{y_b} \beta^{(x_a)}_{y_a}
\end{equation}
for all $(x_a,x_b,y_a,y_b)\in X_A\times X_B\times Y_A\times Y_B$. This shows that for a non-signaling resource
$\omega$ to be realizable by the use of a bipartite quantum
system, with parts corresponding to the Hilbert spaces $\mathcal H_A$ and $\mathcal H_B$, prepared in the state given by the density operator $\rho$, there must exist, for each
$x_a\in X_A$ and $x_b\in X_B$, a positive decomposition
$\left(\beta^{(x_a)}_{y_a}\right)_{y_a\in Y_a}$ of $\rho_B=\tr_A \rho$
and a partition of unity $\left(E^{(x_b)}_{y_b}\right)_{y_b\in Y_b}$
such that (\ref{fromBobspointofview}) holds. In some cases, we can turn this construction the other way around.
\begin{Lemma}\label{existenceofPOVM}
Let $\mathcal H_A$ and $\mathcal H_B$ be separable Hilbert spaces, $\rho$ a density operator on
$\mathcal H_A\otimes \mathcal H_B$,  $\Phi_\rho$ the map
defined by (\ref{defofphi}), $\rho_B=\tr_A\rho$, and
finally $\mathcal K=\rho_B^{1/2}\mathcal B(\mathcal H_B)\rho_B^{1/2}.$
If $\rho$ is pure (i.e., it is an orthogonal projection of rank 1), then there exists a linear map
$\Gamma_\rho:\mathcal K \to \mathcal B(\mathcal H_A)$ such that
\begin{itemize}
\item $\Phi_\rho\circ \Gamma_\rho= {\rm id}_{\mathcal K}$; i.e., $\Gamma_\rho$ is a right-inverse of $\Phi_\rho$,
\item $\Gamma_\rho(K)\geq \mathbf 0$ whenever $K\geq \mathbf 0$; i.e., $\Gamma_\rho$ is a positive map,
\item $\Gamma_\rho(\rho_B)=\mathbf 1$.
\end{itemize}
Hence for every positive decomposition $(\beta_{y})_{y\in Y}$ of $\rho_B$, the formula
$F_y:=\Gamma_\rho(\beta_y)$ defines a POVM for which  
$\Phi_\rho(F_y)=\beta_y$  holds for all $y\in Y$.
\end{Lemma}
\begin{proof}
Suppose $\rho=|\Psi \rangle\langle \Psi|$, where $\Psi\in\mathcal H_A\otimes \mathcal H_B$ is a unit vector.
By the existence of a Schmidt decomposition, we have a  countable  set $S$, an orthonormal system $\left(e^A_n\right)_{n\in S}$ in $\mathcal H_A$, another one $\left(e^B_n\right)_{n\in S}$ in $\mathcal H_B$, and some positive numbers $(\lambda_n)_{n\in S}$ such that
$$\Psi = \sum_{n\in S} \lambda_n \, e^A_n\otimes e^B_n.$$
Moreover, we have that $$\rho_A\equiv \tr_B\rho=
\sum_{n\in S} \lambda_n^2 \, |e^A_n\rangle\langle e^A_n|$$
and similarly, $\rho_B=\sum_{n\in S} \lambda_n^2 \, |e^B_n\rangle\langle e^B_n|$. Let us further consider the partial isometry
$$
V=\sum_{n\in S} |e^A_n\rangle\langle e^B_n|
$$
and  the orthogonal projections  $Q^A=VV^*$ and
$Q^B=V^*V$ onto the
 closures of the subspaces spanned by $\left\{e^A_n|n\in S\right\}$ and
$\left\{e^B_n|n\in S\right\}$, respectively. Finally, we choose
an anti-unitary map $J:\mathcal H_A\to \mathcal H_A$ satisfying
$J e^A_n=e^A_n$ for every $n\in S$, denote $\tilde Z:=J V Z^* V^* J$ for any $Z\in\mathcal B(\mathcal H_B)$, and define $\Gamma_\rho$ by setting
\begin{align*}
\Gamma_\rho(K)= \Gamma_\rho\left(\rho_B^{1/2}Z\rho_B^{1/2}\right):=
\tilde Z + (\tr K)\left(\mathbf 1 -Q^A\right)
\end{align*}
for any $$K=\rho_B^{1/2}Z\rho_B^{1/2}\in
\rho_B^{1/2}\mathcal B(\mathcal H_B)\rho_B^{1/2}=\mathcal K.$$
By the above formula, it is evident that $\Gamma_\rho$ is
well defined (note that $\rho_B^{1/2}Z\rho_B^{1/2}= \rho_B^{1/2}{Z}'\rho_B^{1/2}$ implies
$VZV^*=V{Z}'V^*$), that it is linear (because both the adjoint map and $J$ are anti-linear), that it is a positive map from $\mathcal K$ to $\mathcal H_A$, and  that $\Gamma_\rho(\rho_B)=\mathbf \mathbf 1$.

It is easy to see that
$$\rho\left(\left(\mathbf 1 -Q^A\right)\otimes \mathbf 1\right)=0,$$
and hence that the part $(\tr K) \left(\mathbf 1 -Q^A\right)$ appearing in the definition of $\Gamma_\rho(K)$,
can be ignored when considering the composition $\Phi_\rho\circ \Gamma_\rho$.
 Thus, for any $T\in\mathcal B(\mathcal H_B)$, and for $Z$ and $K$ as before,
we have
\begin{eqnarray}
\nonumber
\tr\Phi_\rho(\Gamma_\rho(K))T = \tr\rho (\Gamma_\rho(K)\otimes T) = \\ 
\nonumber =
\left\langle\Psi,(\tilde Z\otimes T)\Psi\right\rangle =
\\
\nonumber
= \sum_{n,m\in S}
\lambda_n\lambda_m
\left\langle e^A_n\otimes e^B_n,(\tilde Z \otimes T)\left(e^A_m\otimes e^B_m\right)\right\rangle=
\\
\nonumber
= \sum_{n,m\in S}
\lambda_n\lambda_m
\left\langle Z^* e^B_m,e^B_n\right\rangle
\,
\left\langle e^B_n,T e^B_m\right\rangle =
\\
\nonumber
= \sum_{n,m\in S}
\lambda_m
\left\langle e^B_m, Z \rho_B^{1/2}T e^B_m\right\rangle=\\
\nonumber
=\tr \rho_B^{1/2} Z \rho_B^{1/2}B=
\tr KT,
\end{eqnarray}
showing that $\Phi_\rho(\Gamma_\rho(K))=K$ as claimed.
\end{proof}

Suppose now that $\omega$ is realizable by the use of a bipartite quantum system --- with parts corresponding to the  Hilbert spaces $\mathcal H_A$ and $\mathcal H_B$ --- prepared
in the state given by the density operator $\rho$. When defining
$C_n^\omega(X\to Y)$, we needed to consider all protocols involving four different kinds of codings (two on the sender side and two on the receiver side).
It is not difficult to see that all these codings can be incorporated into the choice of partitions / positive operator valued measures, and hence that $\mathcal N\in C_n^\omega(X\to Y)$ if and only if,
 for each $x\in X$ and $r\in [n]$,
there exist a partition of unity $\left(F^{(x)}_{s}\right)_{s \in [n]}$ on $\mathcal H_A$ and a partition of unity $\left(E^{(r)}_{y}\right)_{y \in Y}$ on $\mathcal H_B$ such that
$$
\mathcal N(y|x)= \sum_{r=1}^n \tr\rho\left(F^{(x)}_r\otimes E^{(r)}_y\right)
$$
for all $x\in X$ and $y\in Y$. In particular, if $\rho$
is a density operator corresponding to a {\it maximally entangled state}; i.e., if $$d:=\dim\mathcal H_A= \dim\mathcal H_B <\infty,$$ $\rho$ is pure and $\tr_A\rho=(1/d)\mathbf 1$, then, by Lemma \ref{existenceofPOVM}, the channel $\mathcal N$
is in $C_n^\omega(X\to Y)$ if and only if,
for each $x\in X$ and $r\in [n]$, there exists a positive decomposition  $\left(\beta^{(x)}_s\right)_{s\in [n]}$ of $\mathbf 1/d$ and a partition of unity $(E^r_y)_{y\in Y}$ such that
$$
\mathcal N(y|x) = \sum_{r=1}^n \tr E^r_y \beta^{(x)}_r
$$
for all $x\in X$ and $y\in Y$. In what follows, we will apply the above formula specifically with $n=2$,
and use the notation $E^\pm_y$ and $\beta^{(x)}_\pm$ rather
than $E^r_y$ $(r=1,2)$ and $\beta^{(x)}_r$ $(r=1,2)$.

\section{Main result}  Our goal is to show that a classical bit assisted by a maximally entangled quantum state can be simulated by  two classical bits assisted only by shared randomness.  The proof relies on the method that was used in \cite{F, mi}  to obtain simulation results. We shall need the  following trace inequality.
\begin{Lemma}  For any operators ${\bf 0}\le E^\pm\le\bf 1$ and $\beta_\pm\ge\bf 0$ such that  $\beta_++\beta_-=:\rho_B$ is a  density operator, we have
\begin{equation*}\label{lemma}
\left|\tr E^+E^-\!\rho_B  \right|^2\le
\tr E^+\beta_+ \,+\, \tr E^-\beta_-.
\end{equation*}
\end{Lemma}

\begin{proof}
Set $c_\pm = \tr E^\pm \beta_\pm $ and $t_\pm=\tr\beta_\pm$;
then $c_\pm$ and $t_\pm$ are all nonnegative, and $t_++t_-=1$.
Using the Cauchy--Schwarz inequality $$|\tr AB|^2\le (\tr A^*A) \cdot (\tr B^*B),$$ we have
\begin{eqnarray}
\nonumber
\left|\tr  E^+ E^-\beta_+ \right|^2 =\left|\tr \beta_+^{1/2}E^+E^-\beta_+^{1/2}\right|^2
\le \\ \nonumber \le
\tr\left((E^+)^2 \beta_+\right) \cdot\tr\left((E^-)^2\beta_+\right)
\\
\nonumber
\le \left(\tr E^+\beta_+\right) \cdot \,\tr\beta_+  \;\;\;= \;c_+ t_+,
\end{eqnarray}
and, similarly,
$\left|\tr \beta_-E^+E^-\right|^2\leq c_- t_-$
by interchanging $+$ and $-$ throughout.
Therefore,
\begin{eqnarray}
\nonumber
\left|\tr E^+E^-\!\rho_B \right|^2 =
 \left|\tr E^+E^-\beta_+  \, + \, \tr E^+E^-\beta_-\right|^2\leq
\\
\nonumber
\leq
 \left(|\tr E^+E^- \beta_+| \,  + \, |\tr E^+E^-\beta_-|\right)^2
\\
\nonumber
 \le
  \left(\sqrt{c_+t_+}  \,+ \, \sqrt{c_+t_+}\right)^2.
 \end{eqnarray}
Computing this last square we find that
\begin{eqnarray}
 \nonumber
 \left(\sqrt{c_+t_+}  +  \sqrt{c_+t_+}\right)^2 =\\ \nonumber =
  c_+t_+ + c_-t_- + 2 \,\sqrt{(c_+t_-)(c_-t_+)}\leq
 \\
  \nonumber
  \leq
  c_+t_+ + c_-t_- + 2\, \frac{c_+t_- + c_-t_+}{2} \,=\, c_+ + c_-
 \end{eqnarray}
 by the inequality between the geometric and arithmetic means and the fact that $t_++t_-=1$.  Putting together the last two inequalities, we have $\left|\tr E^+E^-\!\rho_B \right|^2\leq c_++c_-$, 
  as claimed.
\end{proof}

\begin{Th}\label{main}Let $\omega$ be a non-signaling resource realizable by the use of a  bipartite quantum system prepared in a  maximally entangled state. Then $C_2^{\omega}(X\to Y)\subseteq C_4^{SR}(X\to Y)$  for any finite alphabets $X$ and $Y$, i.e.,  a classical bit assisted by  $\omega$  can be simulated by  two classical bits assisted only by  shared randomness.
\end{Th}

\begin{proof}  Let $l=|X|$ and $k=|Y|$. The Theorem is equivalent to the statement that any $k\times l$ matrix $A=(a_{ij})_{i,j}$  with entries $$a_{ij}=\tr E^+_i\beta_+^{(j)} \,+\, \tr E^-_i \beta_-^{(j)},$$ where  $E_i^\pm$ and $\beta_\pm^{(j)}$ are $d\times d$ positive semidefinite matrices with $E^+_1+\dots+E^+_k=E^-_1+\dots+E^-_k=\bf 1$  and $\beta_+^{(j)}+\beta_-^{(j)}={\bf 1}/d$ for all $j\in[l]$, is a convex combination of stochastic matrices with at most four non-zero rows.

For $I=(i_1,i_2, i_3, i_4)\in [k]^4$, put
\begin{equation}
\label{p}
p_I=\frac1{d^2}\left(\tr E_{i_1}^+E_{i_2}^-\right) \left( \tr E_{i_3}^+E_{i_4}^-\right).
\end{equation}

We have $p_I\ge 0$ for all $I$.
Thus, we get a  measure $P$
on $[k]^4$ defined by $P(T)=\sum_{I\in T}p_I$.
Due to the multilinear nature of \eqref{p}  and the assumption that $E_1^\pm$, \dots, $E_k^\pm$ is a partition of unity (POVM), we see that
\[P([k]^4)=\frac1{d^2}\left(\tr({\bf 1}^2)\right) \left(\tr ({\bf 1}^2)\right)=\rm 1,\]
so $P$ is  a probability measure.
Now set $$E^\pm_S:=\sum_{i\in S} E^\pm_i$$
for any $S\subseteq [k]$.
Since  $\mathbf 0\le E^\pm_S\le \bf 1$,  we may apply Lemma \ref{lemma} with $\rho_B= {\bf 1}/d$ to get
\[
P(S^4)=\frac{1}{d^2}\left(\tr E^+_SE^-_S\right)^2\le\tr E^+_S \beta_+^{(j)} \, + \, \tr  E^-_S \beta_-^{(j)}
\]
for all $j$.
The right hand side here is $A_j(S)$, where $A_j$ is the probability measure on $[k]$ given by the numbers $a_{ij}$ $(i\in[k])$; i.e.\! the $j^{\rm th}$ column of the matrix $A$. So we have
\[A_j(S)\ge P\left(S^4\right)\qquad\textrm{ for all }\;S\subseteq [k].\]

Let us connect $I\in [k]^4$ to $i\in [k]$ by an edge if $i$ occurs in $I$. This gives us a bipartite graph. The neighborhood of any set $T\subseteq [k]^4$ is the set $S\subseteq [k]$ of indices occurring in some element of $T$. We always have $T\subseteq S^4$, whence
$$A_j(S)\ge P(S^4)\ge P(T).$$ Thus, by the Supply--Demand Theorem~\cite[2.1.5.\ Corollary]{LoPlu}, and using the fact that both $A_j$ and $P$ are probability measures, there exists a probability measure $P_j$ on $[k]^4\times [k]$ which is supported on the edges of the graph and has marginals $P$ and $A_j$. Whenever $p_I\ne 0$, let $B(I)$ be the $k\times l$ stochastic matrix whose $j$-th column is given by the conditional distribution  $P_j|I$ on $[k]$. 
 Now $B(I)$ has at most four nonzero rows,  and $A=\sum p_I B(I)$, as desired.
\end{proof}

\begin{Rem}
Suppose that our bipartite quantum system is {\it not} in a maximally entangled state, and hence $\rho_B$ is not (necessarily) a multiple of the identity. Still, the above proof could be virtually copied if we had a bilinear, scalar-valued map $D$ satisfying
\begin{itemize}
\item[(i)]  $D(Z_1,Z_2)\geq 0$ whenever $Z_1, Z_2\geq \mathbf 0$,

\item[(ii)] $D(\mathbf 1,\mathbf 1)=1$,

\item[(iii)] $|D(E^+\!,E^-)|^2\, \leq \, \tr E^+\beta_+ \,+ \,\tr E^-\beta_-$ whenever $\mathbf 0 \leq E^\pm\leq \mathbf 1$
and $\rho_B=\beta_+ + \beta_-$ is a positive decomposition of $\rho_B$.
\end{itemize}
Indeed, having such a bilinear map, we could replace (\ref{p}) by setting $$p_I = D\left(E^+_{i_1},E^+_{i_2}\right)\, D\left(E^+_{i_3},E^+_{i_4}\right)$$ and continue the rest of the argument unchanged. Actually, in the proof we {\it did} set $p_I$ to be of the mentioned form; specifically, with $D$ being the bilinear map given by the formula $D(Z_1,Z_2)=(1/d)\tr Z_1 Z_2$.

When $\rho_B$ is not necessarily ${\bf 1}/d$, one could try to replace the previous formula by $D(Z_1,Z_2)=
(\tr Z_1 Z_2 \rho_B + \tr Z_2 Z_1 \rho_B) / 2$. This reduces to the previous one when $\rho_B=\mathbf 1/d$, and it satisfies requirements (ii) and (iii); this latter one follows from Lemma \ref{lemma} and the fact that for self-adjoint operators  $E^\pm$, we have
\begin{align*}|D(E^+,E^-)|^2 = \left({\rm Re}(\tr E^+ E^-\!\rho_B)\right)^2
\leq \\ \leq |\tr E^+ E^-\!\rho_B|^2.\end{align*}
However, this $D$ does not satisfy the positivity condition (i) --- unless of course $\rho_B$ is a multiple of the identity.

Another idea is to try setting $D(Z_1,Z_2)=\tr Z_1\rho_B^{1/2}
Z_2\rho_B^{1/2}$, which again reduces to the formula used in our proof in case $\rho_B$ is a multiple of the identity. The thus defined $D$ is evidently bilinear and satisfies both the positivity (i) and the normalization (ii) requirements. However, examples show that in general it fails to satisfy requirement (iii) --- unless, for example, if $\rho_B$ is
a multiple of a projection.

Having experimented with various candidate formulas, we grew skeptical about the possibility of simultaneously satisfying all listed requirements. Thus, while we still believe that the theorem remains true even if arbitrary entangled states are allowed, we expect the general proof to follow a somewhat different direction.

\end{Rem}

\appendix

\section{The ``two winning, two losing boxes'' game}
\label{appendix}

Let $\rho=|\Psi\rangle\langle \Psi|$, where $\Psi =
\frac{1}{\sqrt{2}}(e_1\otimes e_2 - e_2\otimes e_1)$ and
$(e_1,e_2)$ is the standard basis  of $\mathbb C^2$. Before the game begins, Alice and Bob prepare a pair of quantum bits in the state given by $\rho$; Alice then takes the first, Bob the second quantum bit with herself / himself. Upon learning the positions  $a,b\in \{1,2,3,4\}$ of treasures, Alice performs the measurement corresponding to the $2\times 2$ partition of unity $F^{\{a,b\}}_+$, $F^{\{a,b\}}_-$ and sends the result, a $+$ or a $-$ sign, to Bob via the noiseless one-bit channel. For the specific protocol we want to describe, we will have
$F^{\{1,2\}}_+ = \mathbf 1$, $F^{\{1,2\}}_- = \mathbf 0$
(i.e., in case the treasures are in the first two boxes, Alice will surely send a ``$+$'' to Bob), $F^{\{1,3\}}_\pm = (1/2)(\mathbf 1\pm \sigma_z)$,
$F^{\{2,4\}}_\pm = (1/2)(\mathbf 1\mp \sigma_z)$,
$F^{\{1,4\}}_\pm = (1/2)(\mathbf 1\pm \sigma_x)$,
$F^{\{2,3\}}_\pm = (1/2)(\mathbf 1\mp \sigma_x)$,
where 
 $\sigma_z$ and $\sigma_x$ are two Pauli matrices,
and, finally, $F^{\{3,4\}}_+ = 0$, $F^{\{3,4\}}_- = I$ (so that in case the treasures are in the last two boxes, Alice will surely send a ``$-$'' to Bob).

After receiving the $+$ or $-$ sign from Alice, Bob performs the measurement corresponding to the partition of unity $E^\pm_1$, $E^\pm_2$, $E^\pm_3$, $E^\pm_4$ and chooses the box according to the result.
We will specifically have $E^+_1=(1/2)\left(\mathbf 1- (\sigma_z+\sigma_x)/\sqrt{2}\right)$,
$E^+_2=\mathbf 1-E^+_1$,
$E^+_3=0$, $E^+_4=0$ and $E^-_1=0$,
$E^-_2=0$,
$E^-_3=(1/2)\left(\mathbf 1 + (\sigma_z-\sigma_x)/\sqrt{2}\right)$,
$E^-_4=\mathbf 1-E^-_3$.

As $E^+_3=E^+_4=0$ and likewise,
$E^-_1=E^-_2=0$, Bob will always choose one of the first two boxes if he receives a $+$, and one of the last two
boxes if he receives a $-$ sign. Hence if the two treasure boxes are either the first
two or the last two, they will win with certainty. On the other hand, if the treasures are e.g.\! in boxes $1$ and $3$, then they win with probability
\begin{align*}
\tr\rho\left(F^{\{1,3\}}_+\otimes (E^+_1 + E^+_3)\right)+\\ +\tr\rho\left(F^{\{1,3\}}_-\otimes (E^-_1 + E^-_3)\right),
\end{align*}
which, after substitution, turns out to be  $$\frac{1}{2}+\frac{1}{4}\sqrt{2}=\left(2+\sqrt{2}\right)/4.$$ 
It turns out that all other cases  result in the same probability of success, yielding the claimed overall winning probability of $\left(4+\sqrt{2}\right)/6$.
We finish the discussion of  this example by pointing out that all listed measurements are either trivial or projective; the entire protocol can be easily realized experimentally using e.g.\! a pair of spin-half particles prepared in the zero-total-spin state and spin measurements performed on individual particles.

\end{document}